\documentclass{article}

\usepackage[utf8]{inputenc} % allow utf-8 input
\usepackage[T1]{fontenc}    % use 8-bit T1 fonts
\usepackage{hyperref}       % hyperlinks
\usepackage{url}            % simple URL typesetting
\usepackage{booktabs}       % professional-quality tables
\usepackage{amsfonts}       % blackboard math symbols
\usepackage{nicefrac}       % compact symbols for 1/2, etc.
\usepackage{microtype}      % microtypography
\usepackage{lipsum}
\usepackage{fancyhdr}       % header
\usepackage{graphicx}       % graphics
\usepackage{tikz}
\usepackage{amsthm}
\usepackage{xcolor}
\newtheorem{theorem}{Theorem}
\newtheorem{lemma}{Lemma}
\newtheorem{corollary}{Corollary}
\newtheorem{observation}{Observation}
\newtheorem{definition}{Definition}

\author{
  \textbf{Vinicius L. do Forte}\\
  Universidade Federal Rural do Rio de Janeiro\\ Rio de Janeiro Brazil\\
  \texttt{vlforte@gmail.com} \\
  \and
   \textbf{Min Chih Lin} \\
  Instituto de Cálculo and Departamento de Computación\\ Universidad de Buenos Aires\\
  Buenos Aires, Argentina\\
  \texttt{oscarlin@dc.uba.ar} \\
  \and
  \textbf{Abilio Lucena} \\
  Programa de Engenharia de Sistemas e Computação\\
  Universidade Federal do Rio de Janeiro\\
  Rio de Janeiro, Brazil\\
  \texttt{abiliolucena@gmail.com} \\
  \and
  \textbf{Nelson Maculan} \\
  Programa de Engenharia de Sistemas e Computação\\ Universidade Federal do Rio de Janeiro\\ Rio de Janeiro, Brazil\\
  \texttt{maculan@cos.ufrj.br} \\
  \and 
 \textbf{Veronica A. Moyano} \\
  Instituto de Ciencias\\ Universidad Nacional de General Sarmiento\\ Buenos Aires, Argentina\\
  \texttt{vmoyano@ic.fcen.uba.ar} \\
  \and 
  \textbf{Jayme L. Szwarcfiter} \\
  Programa de Engenharia de Sistemas e Computação\\ Universidade Federal do Rio de Janeiro\\ Rio de Janeiro, Brazil \\ Instituto de Matemática e Estatística\\ Universidade do Estado do Rio de Janeiro\\ Rio de Janeiro, Brazil\\}
\begin{document}
%% Title
\title{Graphs whose vertices of degree at least 2 lie in a triangle}

\maketitle

%\begin{center}
  %\textbf{Vinicius L. do Forte}\\
  %Universidade Federal Rural do Rio de Janeiro, Brazil\\
  %\texttt{vlforte@gmail.com} \\
  %
   %\textbf{Min Chih Lin} \\
  %Instituto de Cálculo and Departamento de Computación, Universidad de Buenos Aires, Buenos Aires, Argentina\\
  %\texttt{oscarlin@dc.uba.ar} \\
  %
  %\textbf{Abilio Lucena} \\
  %Programa de Engenharia de Sistemas e Computação, Universidade Federal do Rio de Janeiro, Rio de Janeiro, Brazil\\
  %\texttt{abiliolucena@gmail.com} \\
  %
  %\textbf{Nelson Maculan} \\
  %Programa de Engenharia de Sistemas e Computação, Universidade Federal do Rio de Janeiro, Rio de Janeiro, Brazil\\
  %\texttt{maculan@cos.ufrj.br} \\
   %
 %\textbf{Veronica A. Moyano} \\
  %Instituto de Ciencias, Universidad Nacional de General Sarmiento, Buenos Aires, Argentina\\
  %\texttt{vmoyano@ic.fcen.uba.ar} \\
   %
  %\textbf{Jayme L. Szwarcfiter} \\
  %Programa de Engenharia de Sistemas e Computação, Universidade Federal do Rio de Janeiro, Rio de Janeiro, Brazil and Instituto de Matemática e Estatística, Universidade do Estado do Rio de Janeiro, Rio de Janeiro, Brazil\\
  %\texttt{jayme@cos.ufrj.br}
%\end{center}

\begin{abstract}
A pendant vertex is one of degree one and an isolated vertex has degree zero. A neighborhood star-free (NSF for short) graph is one in which every vertex is contained in a triangle except pendant vertices and isolated vertices. This class has been considered before for several contexts. In the present paper, we study the complexity of the dominating induced matching (DIM) problem \textcolor{black}{and the perfect edge domination (PED) problem for NSF graphs. We prove the corresponding decision problems are NP-Complete for several of its subclasses. As an added value of this study, we have shown three connected variants of planar positive 1in3SAT are also NP-Complete. Since these variants are more basic in complexity theory context than many graph problems, these results can be useful to prove that other problems are NP-Complete.}
\end{abstract}

\textbf{keyword: }
Neighborhood star-free graphs \and Triangle \and Efficient edge domination \and Perfect edge domination, Dominating induced matching \and Connected planar positive 1in3SAT \and Connected cubic planar positive 1in3SAT \and Connected subcubic planar $C_4$-free positive 1in3SAT

%\keywords{neighborhood-star-free graphs \and triangle \and efficient edge domination \and perfect edge domination \and dominating induced matching \and crickets \and cubic planar monotone 1in3SAT \and planar $C_4$-free positive 1in3SAT}

\section{Introduction}
A graph with vertex set $V$ and edge set $E$ is denoted by $G=(V,E)$, and the number of vertices and edges is $|V|=n$ and $|E|=m$, respectively. The open and closed neighborhoods of $v \in V$ are denoted by $N(v)$ and $N[v]=N(v)\cup\{v\}$, respectively. Write $d(v)=|N(v)|$ for the degree of $v$. Let $e=(u,w)\in E$ an edge, denote by $d(e)=d(u,w)=|N(u)\cap N(w)|$ the degree of $e$. For $V' \subseteq V$, $G[V']$ represents the subgraph of $G$ induced by $V'$. A \textit{\textbf{star}} is the bipartite graph $K_{1,t}$, for some $t \geq 1$. 

We say that $G=(V,E)$ is a \textit{\textbf{neighborhood star-free}} graph, NSF graph for short, if for every vertex $v\in V$ with degree at least 2, $G[N[v]]$,  is not a star. In other words, every vertex $v\in V$ is contained in some triangle, for $d(v) \geq 2$. Such a class of graphs seems natural, in some different ways. For instance, to model human relationships, as the common case where every individual is acquainted with at least a pair of individuals, which are themselves acquainted.  In this paper, we explore the use of this class in the problems of perfect and efficient edge domination of graphs. We remark that NSF graphs are not necessarily closed under induced subgraphs \textcolor{black}{(adding 2 new vertices of degree $n+1$ to any $n$-vertex graph makes it an NSF graph)}. In this paper, we focus on NSF graphs and this class has been considered before on various contexts.

\textcolor{black}{
In \cite{DBLP:journals/dm/ChartrandHHZ03}, the authors investigate the $F$-domination number for all 2-stratified graphs $F$ of order at most 3. Particularly, for two 2-stratified graphs $K_3$, they consider (i) graphs in which every vertex is in a triangle and (ii) graphs in which every edge is in a triangle; both are NSF graph subclasses. A graph $G$ is \textit{\textbf{2-stratified graph}} if its vertex set is partitioned into two non-empty color classes: red and blue. Let $F$ be a 2-stratified graph rooted at some blue vertex $v$. The \textit{\textbf{$F$-domination number $\gamma_F(G)$}} of a graph $G$ is the minimum number of red vertices of $G$ in a red–blue coloring of
the vertices of $G$ such that every blue vertex $v$ of $G$ belongs to a copy of $F$ rooted at $v$.}

\textcolor{black}{In a very recent paper \cite{Brandenburg_2023}, it is shown that every vertex of any optimal IC-planar graph is in a triangle. An $n$-vertex graph $G$ is \textit{\textbf{IC-Planar}} if it has a drawing in the plane such that each edge is crossed by at most one edge and every vertex is incident to at most one crossed edge. Such IC-planar graph $G$ is \textit{\textbf{optimal}} if every $n$-vertex graph $G'$ with more edges than $G$ is non-IC-Planar.}

\textcolor{black}{Also, any $k$-tree for fixed $k\geq 2$ is an NSF graph. A complete graph on $k$ vertices is a $k$-tree, fixed $k$. If $G$ is a $k$-tree, fixed $k$ and $C$ is a $k$-clique of $G$, then the graph formed by adding a new vertex to $G$, and making it adjacent to all vertices of $C$ is a $k$-tree, fixed $k$. Several characterizations of $k$-trees are given in \cite{ROSE1974317}. As maximal outerplanar graphs are $2$-trees \cite{MARKENZON2006818}, they are NSF graphs.
}

Most of our proposed results on this paper are on \textcolor{black}{edge} domination problems. For a graph $G=(V,E)$, every edge $e \in E$ \textit{\textbf{dominates}} itself and all edges adjacent to $e$. A subset $E' \subseteq E$ is an \textit{\textbf{edge dominator}} of $G$, if $E'$  dominates every edge $e \in E$ of $G$. In particular, if each edge 
$e \in E$ is dominated exactly by one edge then $E'$ is called an \textit{\textbf{efficient edge dominating set (EED)}}. On the other hand, if we relax the definition, and let each of $e \in E \setminus E'$ be dominated exactly by one edge then $E'$ is called a \textit{\textbf{perfect edge dominating set (PED)}}. It follows from the definitions that every EED is also a PED. Observe that a graph does not necessarily contain an EED \textcolor{black}{and in case of having several EEDs, they all have the same number of edges}. However, every graph always contains a PED. In particular, the entire edge set $E$ is always a PED, called \textit{\textbf{trivial}}. Every PED which is not trivial and not an EED is called \textit{\textbf{proper}}. \textcolor{black}{As many domination problems, the interest is to find a PED (EED) of smaller size or weight if the graph has weights assigned to its edges. All these problem variants are NP-hard for graphs in general. In particular, determining whether or not a graph has an EED is NP-Complete\cite{Gr-Sl-Sh-Ho-1993}. Below is a summary table of the complexity status of this problem by different graph classes and their references. Also in \cite{CARDOSO2011521}, the authors define a limite class $\mathcal{S}=\bigcap \mathcal{S}_k$ (see Table \ref{tab:table1}) for EED problem and they claim that unless $P=NP$, the problem is solvable in polynomial time in an hereditary graph class $X=Free(M)$ only if $X$ excludes all classes $\mathcal{S}_k$, i.e., only if $M \cap \mathcal{S}_k\ne\emptyset$, for each $k$.
} 
%\textcolor{blue}{Agregar\'ia aqu\'i, las subclases que son NP-Completes y las que son polinomiales para cada problema con sus referencias. Tambi\'en hacer lo mismo con los resultados de dicotom\'ia. Finalmente, mencionar que todos estos estudios se centran casi exclusivamente en subclases de grafos que son hereditarias. Por eso es interesante estudiar esta clase no hereditaria.}
\begin{table}[h!]
  \begin{center}
    \label{tab:table1}
    \begin{tabular}{c|c|c} % <-- Alignments: 1st column left, 2nd middle and 3rd right, with vertical lines in between
      \textbf{Graph Class} & \textbf{Complexity} & \textbf{Ref.}\\
      \hline
			cubic graphs & NP-Complete & \cite{Kratochvil94}\\
      $r$-regular graphs for fixed $r\geq 3$ & NP-Complete & \cite{CardosoCDS08}\\
      bipartite graphs & NP-Complete & \cite{CHINLUNGLU1998203}\\
      planar bipartite graphs & NP-Complete & \cite{Lu-Ko-Tang-2002}\\
      planar bipartite graphs with degree $\leq 3$ & NP-Complete & \cite{Br-2010}\\
      $\mathcal{S}_k=(C_3,\cdots,C_k,H_1,\cdots,H_k)$-free bipartite graphs with degree $\leq 3$  & NP-Complete & \cite{CARDOSO2011521}\\
			circular-arc graphs & Linear & \cite{LIN2015}\\
			claw-free graphs & $O(n)$ & \cite{CARDOSO2011521,LIN2014524}\\
			convex graphs & Linear & \cite{CARDOSO2011521,Br-2010}\\
			biconvex graphs & $O(n)$ & \cite{Br-2010,LIN2014524}\\ 
			bounded clique-width graphs & Polynomial & \cite{CARDOSO2011521}\\
			bipartite permutation  graphs & $O(n)$ & \cite{CHINLUNGLU1998203,LIN2014524}\\ 
			generalized series–parallel graphs & Linear & \cite{Lu-Ko-Tang-2002}\\
			chordal graphs & $O(n)$ & \cite{Lu-Ko-Tang-2002,LIN2014524}\\
			dually chordal graphs & $O(n)$ & \cite{DBLP:conf/isaac/BrandstadtLR12,LIN2014524}\\ 
			hole-free graphs & Polynomial & \cite{Br-2010}\\
			chordal bipartite graphs & Linear & \cite{Br-2010}\\
			$P_7$-free graphs & Linear & \cite{Br-Mo-2014}\\
			$P_8$-free graphs & Polynomial & \cite{Br-Mo-2017}\\
			$P_9$-free graphs & Polynomial & \cite{Br-Mo-2022}\\
			$S_{1,2,2}$-free graphs & Polynomial & \cite{DBLP:journals/jda/KorpelainenLP14}\\
			$S_{1,2,3}$-free graphs & Polynomial & \cite{DBLP:journals/jda/KorpelainenLP14}\\
			$S_{1,2,4}$-free graphs & Polynomial & \cite{Br-Mo-2020}\\
			$S_{2,2,2}$-free graphs & Polynomial & \cite{He-Lo-Ri_Za-deW-2018}\\
			$S_{2,2,3}$-free graphs & Polynomial & \cite{Br-Mo-2020-2}\\
			$S_{1,1,5}$-free graphs & Polynomial & \cite{Br-Mo-2020-3}\\
    \end{tabular}
    \caption{Complexity Status for Existence of EED problem by graph classes}
  \end{center}
\end{table}

In contrast, there is less research done on the PED problem.

\begin{table}[h!]
  \begin{center}
    \label{tab:table2}
    \begin{tabular}{c|c|c} % <-- Alignments: 1st column left, 2nd middle and 3rd right, with vertical lines in between
      \textbf{Graph Class} & \textbf{Complexity} & \textbf{Ref.}\\
      \hline
      bipartite graphs & NP-Complete & \cite{Lu-Ko-Tang-2002}\\
      claw-free graphs with degree $\leq 3$ & NP-Complete & \cite{Li-Lo-Mo-Sz}\\
			$r$-regular graphs for fixed $r\geq 3$ & NP-Complete & \cite{Li-Lo-Mo-Sz}\\
			graphs with degree $\leq r$ and girth $\geq k$ for fixed $r,k\geq 3$ & NP-Complete & \cite{Li-Lo-Mo-Sz}\\
			generalized series–parallel graphs & Linear & \cite{Lu-Ko-Tang-2002}\\
			chordal graphs & Linear & \cite{Lu-Ko-Tang-2002}\\
 			circular-arc graphs & Linear & \cite{LIN2015}\\
      cubic claw-free graphs & $O(n)$ & \cite{Li-Lo-Mo-Sz}\\
			$P_5$-free graphs & Linear & \cite{Li-Lo-Mo-Sz}\\
    \end{tabular}
    \caption{Complexity Status for the PED problem by graph classes}
  \end{center}
\end{table}

Also, a dichotomy theorem is given in \cite{Li-Lo-Mo-Sz} which establishes the complexity of the PED problem for every graph class $\mathcal{G}(H,d)$ which is $H$-free graphs with degree $\leq d$ for some fixed $d\geq 3$ and $H$ a graph. The problem is polynomial time solvable for graphs in $\mathcal{G}(H,d)$ if $H$ is linear forest (union of induced paths) and NP-Complete otherwise.

We can see that most of the studies done on these problems so far focus their attention on the hereditary graph classes. It would be of great interest to study them for a non-hereditary graph class like NSF graphs.

In \cite{Li-Lo-Mo-Sz}, the authors proved that \textcolor{black}{connected} NSF graphs do not have any proper perfect dominating set. The only possible PEDs of these graphs are (i) the trivial PED, and (ii) EEDs. \textcolor{black}{We give a brief idea of this proof. Given a perfect edge dominating set $E'\subseteq E$ of a connected graph $G=(V,E)$, we can classify each vertex of $V$ in three types according to the number of its incident edges in $E'$: white, 0 edges; yellow, 1 edge; black, 2 or more edges. The following statements are true.
\begin{itemize}
\item Black vertices have no white neighbors.
\item If  $(u,v)\in E$, then $(u,v)\in E'$ if and only if $u$ and $v$ are both non-white vertices.
\item Any neighbor of a white vertex is yellow.
\item Any yellow vertex has exactly one non-white neighbor.
\item For any triangle $T$ of $G$, $T$ has either all 3 black vertices or 2 of them are yellow and the third one is white.
\item Given a black vertex $u$ and some of its neighbors $v$, if $v$ is in some triangle then $v$ is also black.
\item If $|V|=1$, $E'$ is trivial PED.
\item If $|V|\geq 2$, $E'$ is trivial PED if and only if there is no white vertex.
\item $E'$ is EED if and only if there is no black vertex.
\item If $E'$ is proper PED ($E'$ is not trivial PED nor EED) then there are some black vertex and some white vertex. Clearly, every degree one vertex can not be black. There is some black vertex $u\in V$ with degree at least degree 2. If $G$ is an NSF graph, then $v$ is in some triangle $T$ of $G$ and all vertices of degree 2 or more are black since for each vertex $v$ of degree at least 2, there is a path connecting $u$ with it. Every vertex in this path is in some triangle. Consequently, all degree one vertices are yellow and there is no white vertex which is a contradiction. Then $E'$ is not proper PED or $G$ is not an NSF graph.
\end{itemize}
}

In this article, we show that deciding if a \textcolor{black}{connected} NSF
graph contains an EED is an NP-Complete problem. In fact, we describe  NP-Completeness proofs for several subclasses of \textcolor{black}{connected} NSF graphs. %In addition, we give a polynomial-time algorithm to solve the %minimum weight 
%EED problem for NSF graphs which are cricket-free, using a reduction motivated by that for claw-free graphs in \cite{CARDOSO2011521,LIN2014524}.

We briefly describe some known results related to the present paper.

There is a concept of bi-coloring associated to an EED set \cite{CARDOSO2011521} as follows. Given a graph $G=(V,E)$ and an EED $E' \subseteq E$ of $G$, vertices of $G$ can be classified into two color classes.

\begin{description}
\item[black vertices] if they have exactly one incident edge of $E'$. We denote this subset of vertices by $B$.
\item[white vertices] if they do not have any incident edge of $E'$. We denote this subset of vertices by $W$.
\end{description}

We call $(B,W)$ the \textit{\textbf{bi-coloring associated}} to $E'$.

\begin{observation}
The subset  $B$ of $G$ induces in $G$ the EED $E'$. Moreover, $E'$ is an induced matching of $G$.
\end{observation}

For the above reason, an EED is also called  \textit{\textbf{dominating induced matching (DIM)}} of $G$. From now on, we will use either the terms EED or DIM, with the same meaning.

The following rules are useful to verify if a bi-coloring of a graph $G$ is associated to some DIM of $G$. In the positive case, we say that it is a \textit{\textbf{valid}} bi-coloring of $G$. 
\begin{theorem}\label{check-rulesOri} \cite{CARDOSO2011521}
Given a graph $G$ and a bi-coloring $(B,W)$ of $G$. Then  $(B,W)$ is valid if and only if it verifies the following conditions.
\begin{itemize}
\item $G[B]$ is 1-regular which means every black vertex has exactly one black neighbor.
\item $W$ is an independent set. That is, all neighbors of white vertices have black color.   
\end{itemize}
\end{theorem}

\textcolor{black}{As in \cite{LIN2017}, assigning one of the two possible colors to vertices of $G$ is called a \textit{\textbf{coloring}} of $G$. A coloring is \textit{\textbf{partial}} if only part of the vertices of $G$ has been assigned colors, otherwise is \textit{\textbf{total}}}. 

\begin{definition}\label{valid partial coloring} \cite{LIN2017}
A partial coloring of $G$ is \textit{\textbf{valid}} if verifies the following conditions.
\begin{itemize}
\item White vertices form an independent set.
\item Every black vertex has at most one black neighbor.
\item Every black vertex without black neighbors has an uncolored neighbor.
\end{itemize}
\end{definition}
It is clear that any invalid partial coloring cannot be extended to a valid total coloring.

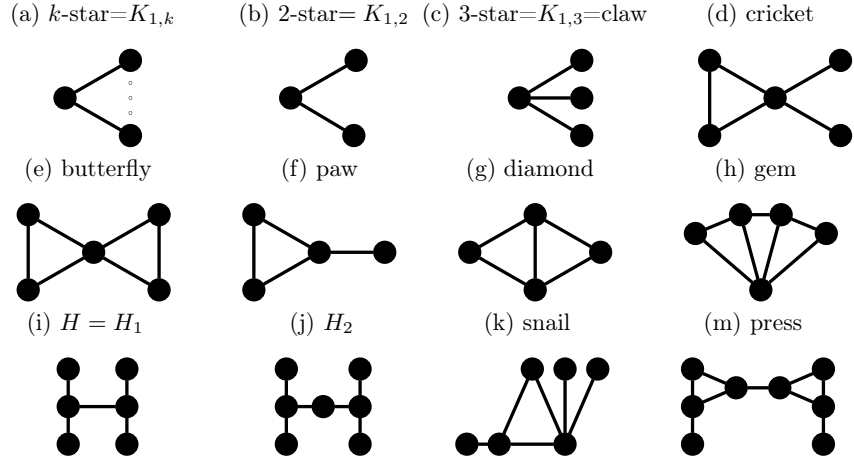
\begin{figure}
\centering
%%%%%Gadgets
\begin{tikzpicture}[v/.style={circle, scale=0.8, draw=black, fill},q/.style={circle, scale=0.01, draw=black, fill},%p/.style={rectangle, scale=0.005, draw=black, fill},
line width=1.25pt]
% Star
\node (0) [v] at (-0.44,0.5){};
\node (1) [q] at (0.43,0.3){};
\node (2) [q] at (0.43,0.5){};
\node (5) [q] at (0.43,0.7){};
\node (3) [v] at (0.43,0){};
\node (4) [v] at (0.43,1){};
%\node (6) [p] at (0.0,0.5){};
%\node (7) [p] at (0.14,0.5){};
%\node (8) [p] at (0.28,0.5){};
%\node (9) [p] at (-0.14,0.5){};
%\node (10) [p] at (-0.28,0.5){};
\draw[] (3) -- (0) -- (4);
\node (a)[scale=0.9] at (0-0.075,1.6){(a) $k$-star=$K_{1,k}$};
% P_2
\node (0c) [v] at (-0.44+3,0.5){};
\node (3c) [v] at (0.4+3,0){};
\node (4c) [v] at (0.43+3,1){};
\draw[] (3c) -- (0c) -- (4c);
\node (a)[scale=0.9] at (3,1.6){(b) $2$-star$=K_{1,2}$};
% claw
\node (0d) [v] at (-0.4+6,0.5){};
\node (2d) [v] at (0.43+6,0.5){};
\node (3d) [v] at (0.43+6,0){};
\node (4d) [v] at (0.43+6,1){};
\draw[] (3d) -- (0d) -- (4d);
\draw[] (0d) -- (2d);
\node (c)[scale=0.9] at (5.8,1.6){(c) $3$-star=$K_{1,3}$=claw};
% claw
%\node (0b) [v] at (-0.4+9,0.5){};
%\node (2b) [v] at (0.43+9,0.33){};
%\node (5b) [v] at (0.43+9,0.67){};
%\node (3b) [v] at (0.43+9,0){};
%\node (4b) [v] at (0.43+9,1){};
%\draw[] (3b) -- (0b) -- (4b);
%\draw[] (5b) -- (0b) -- (2b);
%\node (c)[scale=0.9] at (5.8+3,1.6){(d) $4$-star=$K_{1,4}$};

\node (0b) [v] at (0+9,0.5){};
\node (1b) [v] at (-0.87+9,0){};
\node (2b) [v] at (-0.87+9,1){};
\node (3b) [v] at (0.87+9,0){};
\node (4b) [v] at (0.87+9,1){};
\draw[] (0b) -- (1b) -- (2b) -- (0b);
\draw[] (3b) -- (0b) -- (4b);
\node (c)[scale=0.9] at (5.8+3,1.6){(d) cricket};

\end{tikzpicture}
%%%%%Gadgets
\begin{tikzpicture}[v/.style={circle, scale=0.8, draw=black, fill},
line width=1.25pt]
% Buterfly
\node (0) [v] at (0,0.5){};
\node (1) [v] at (-0.87,0){};
\node (2) [v] at (-0.87,1){};
\node (3) [v] at (0.87,0){};
\node (4) [v] at (0.87,1){};
\draw[] (0) -- (1) -- (2) -- (0);
\draw[] (3) -- (0) -- (4) -- (3);
\node (a)[scale=0.9] at (0-0.075,1.6){(e) butterfly};
% Paw
\node (0) [v] at (0+3,0.5){};
\node (1) [v] at (-0.87+3,0){};
\node (2) [v] at (-0.87+3,1){};
\node (3) [v] at (0.87+3,0.5){};

\draw[] (0) -- (1) -- (2) -- (0);
\draw[] (3) -- (0);
\node (a)[scale=0.9] at (3,1.6){(f) paw};
%%%Diamond
\node (0c) [v] at (6-1,0.5){};
\node (1c) [v] at (6.87-1,1){};
\node (2c) [v] at (6.87-1,0){};
\node (3c) [v] at (7.74-1,0.5){};
\node (c)[scale=0.9] at (5.8,1.6){(g) diamond};
\draw[] (0c) -- (1c) -- (2c) -- (0c);
\draw[] (1c) -- (3c) -- (2c);
%%%Gem
\node (0g) [v] at (8,0.75){};
\node (1g) [v] at (8.6,1){};
\node (2g) [v] at (9.14,1){};
\node (3g) [v] at (9.74,0.75){};
\node (4g) [v] at (8.87,0){};
\node (c)[scale=0.9] at (5.8+3,1.6){(h) gem};
\draw[] (0g) -- (1g) -- (2g) -- (3g);
\draw[] (0g) -- (4g) -- (1g);
\draw[] (2g) -- (4g) -- (3g);

\end{tikzpicture}
%%%%%Gadgets
\begin{tikzpicture}[v/.style={circle, scale=0.8, draw=black, fill},
line width=1.25pt]
% H
\node (0) [v] at (0-0.3,0){};
\node (1) [v] at (0-0.3,0.5){};
\node (2) [v] at (0-0.3,1){};
\node (3) [v] at (0.74-0.26,0){};
\node (4) [v] at (0.74-0.26,0.5){};
\node (5) [v] at (0.74-0.26,1){};
\draw[] (0) -- (1) -- (2);
\draw[] (3) -- (4) -- (5);
\draw[] (1) -- (4);
\node (a)[scale=0.9] at (0-0.075,1.6){(i) $H=H_1$};
% H2
\node (f0) [v] at (0-0.3+3-.1,0){};
\node (f1) [v] at (0-0.3+3-.1,0.5){};
\node (f2) [v] at (0-0.3+3-.1,1){};
\node (f3) [v] at (0.74-0.26+3+.1,0){};
\node (f4) [v] at (0.74-0.26+3+.1,0.5){};
\node (f5) [v] at (0.74-0.26+3+.1,1){};
\node (f6) [v] at (0.09+3,0.5){};

\draw[] (f0) -- (f1) -- (f2);
\draw[] (f3) -- (f4) -- (f5);
\draw[] (f1) -- (f6) -- (f4);

\node (a)[scale=0.9] at (3+.07,1.6){(j) $H_2$};
%%%Snail
\node (0g) [v] at (5,0){};
\node (1g) [v] at (5.43,0){};
\node (2g) [v] at (5.87,1){};
\node (3g) [v] at (6.305,0){};
\node (4g) [v] at (6.305,1){};
\node (5g) [v] at (6.74,1){};
\node (c)[scale=0.9] at (5.8,1.6){(k) snail};
\draw[] (0g) -- (1g) -- (2g) -- (3g) -- (4g);
\draw[] (1g) -- (3g) -- (5g);
%%%Gem
\node (0h) [v] at (8,1){};
\node (1h) [v] at (8,0.5){};
\node (2h) [v] at (8,0.0){};
\node (3h) [v] at (8.58,0.75){};
\node (4h) [v] at (9.16,0.75){};
\node (5h) [v] at (9.74,1){};
\node (6h) [v] at (9.74,0.5){};
\node (7h) [v] at (9.74,0.0){};
\node (c)[scale=0.9] at (5.8+3,1.6){(m) press};
\draw[] (2h) -- (1h) -- (0h) -- (3h) --(4h) --(5h) -- (6h) -- (7h);
\draw[] (1h) -- (3h);
\draw[] (4h) -- (6h);

\end{tikzpicture}

\caption{Some graphs of interests}
\label{Fig0}

\end{figure}
%%%

\begin{figure}
\centering

%%%%%Gadgets
\begin{tikzpicture}[v/.style={circle, scale=0.8, draw=black, fill},
line width=1.25pt]
% K_4=W_3
\node (0k) [v] at (0+0.45,0.5){};
\node (1k) [v] at (-0.87+0.45,0){};
\node (2k) [v] at (-0.87+0.45,1){};
\node (3k) [v] at (-.56+0.45,0.5){};

\draw[] (0k) -- (1k) -- (2k) -- (0k) -- (3k);
\draw[] (1k) -- (3k) -- (2k);
\node (a)[scale=0.9] at (0-0.075,1.6){(a) $K_4=W_3$};
% W_4
\node (0) [v] at (0+3,0.5){};
\node (1) [v] at (-0.87+3,0){};
\node (2) [v] at (-0.87+3,1){};
\node (3) [v] at (0.87+3,0){};
\node (4) [v] at (0.87+3,1){};
\draw[] (3) -- (1) -- (2) -- (4) -- (3);
\draw[] (1) -- (0) -- (2);
\draw[] (3) -- (0) -- (4);
\node (a)[scale=0.9] at (3,1.6){(b) $W_4$};

%%% W_5
\node (0c) [v] at (6-1,0.65){};
\node (1c) [v] at (6.87-1,1){};
\node (2c) [v] at (6+0.5-1,0){};
\node (4c) [v] at (7.74-0.5-1,0){};
\node (3c) [v] at (7.74-1,0.65){};
\node (5c) [v] at (6.87-1,0.5){};
\node (c)[scale=0.9] at (5.8,1.6){(c) $W_5$};
\draw[] (0c) -- (1c) -- (3c) -- (4c) -- (2c) -- (0c) -- (5c) -- (4c);
\draw[] (2c) -- (5c) -- (3c);
\draw[] (1c) -- (5c);

\end{tikzpicture}

\caption{Some wheels}
\label{FigWheel}
\end{figure}
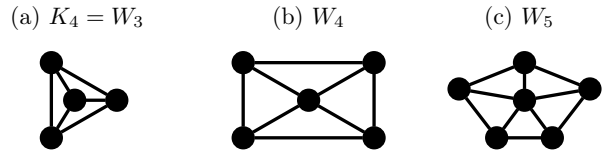
%%%

For the next definitions, we need some more basic terminology.

A vertex is \textit{\textbf{universal}} if it is neighbor of all other vertices \textcolor{black}{of the graph}. A vertex is \textit{\textbf{simplicial}} if the subgraph induced by its neighborhood is complete. 
\textcolor{black}{Next we define some of the graphs depicted in Figure \ref{Fig0}.} 
A \textit{\textbf{butterfly}} is an NSF graph that consists of exactly two triangles sharing a single vertex. 
We call this vertex its \textit{\textbf{central}} vertex. 
If we remove any vertex except the central vertex, the resulting graph is a \textit{\textbf{paw}}. 
A \textit{\textbf{pendant}} or \textit{\textbf{leaf}} vertex is a vertex of degree one and its unique neighbor is a \textit{\textbf{preleaf}} vertex. $P_k$ is a connected graph with exactly $k-2$ vertices of degree 2 and two leaves. $C_k$ is a connected graph of $k$ vertices of degree 2. 
\textcolor{black}{A \textit{\textbf{wheel}}} $W_k$ is a $C_k$ plus a universal vertex (see Figure \ref{FigWheel}). 
A \textit{\textbf{diamond}} is the resulting graph after deleting an edge from a $K_4$. 
A \textit{\textbf{gem}} is \textcolor{black}{a $P_4$} plus a universal vertex.
An $H$ graph is two copies of $P_3$ plus an edge connecting both middle vertices of $P_3$'s. The latter edge is the \textit{\textbf{middle}} edge of $H$.
$H_k$ is obtained by replacing the middle edge $(v,u)$ of $H$ by a path $P_{k+1}$ that connects $v$ and $u$. Exactly $k-1$ new vertices of degree 2 are added.
A \textit{\textbf{snail}} is obtained by \textcolor{black}{adding 3 pendant vertices to a triangle, in such a way that two of them become adjacent to the same vertex of the triangle.}
A \textit{\textbf{press}} graph consists of two copies of paw plus an edge connecting a degree 2 vertex of each paw.
\\
There are some structures and properties related to valid bi-colorings.
\begin{theorem}\label{color-rules} \cite{CARDOSO2011521}
Given a graph $G$ and any valid bi-coloring $(B,W)$ of $G$ 
\begin{itemize}
\item Any triangle of $G$ has exactly one white vertex.
\item The central vertex of an induced butterfly $F$ of $G$ is the unique white vertex of $F$.
\item The odd degree %(in the paw) 
vertices in an induced paw of $G$ have different colors.
\item In an induced diamond, the vertices of degree %(in the diamond)
2 are white and the other two vertices are black.
\end{itemize}
\end{theorem}

There is some more bibliography to add to the already vast literature \cite{Br2018,DBLP:reference/algo/BrandstadtN16,LIN2017,Moyano-2017-thesis,Nevries-2014-thesis,Xiao-agamochi} on dominating induced matchings. As we have seen before the papers on perfect edge domination are less frequent. There is a paper \cite{doFor-Vi-Lin-Lu-Ma-Mo-Sz-2020} where the authors describe ILP formulations for the PED problem, together with some experimental results. 

\section{NP-Completeness results}

We consider in this work five variants of the 1in3SAT problem, two of them are well-known NP-Complete problems, and we prove that the other three are NP-Complete. Then, we reduce them to the existence of DIMs on some subclasses of connected NSF graphs.
\\

Given a formula $F$ in conjunctive normal form (CNF), its associated graph $G(F)$ is a bipartite graph where one of the classes of the bipartition corresponds to the clauses and the other \textcolor{black}{one} to the variables. We follow the convention of previous articles in this field, using rectangles to represent vertices corresponding to clauses and circles for those \textcolor{black}{that are associated} to variables. In $G(F)$, two vertices are neighbors if one of them corresponds to a variable $x_j$ (circle), the other one corresponds to a clause $C_i$ (rectangle), and $x_j$ or $\neg x_j$ is a literal of $C_i$. As we work \textcolor{black}{with 1in3SAT, every} clause of $F$ has exactly three different literals which implies that every clause (rectangle) vertex has degree 3 in $G(F)$. We say that an assignment is \textit{\textbf{valid}} if in every clause there is exactly one true literal. Two formulas $F$ and $F'$ in \textcolor{black}{CNF are equivalent} if there is a valid assignment for $F$ if only if there is a valid assignment for $F'$. A formula $F$ is \textit{\textbf{positive}} if all literals in $F$ are positive.
 
On the other hand, it is highly probable that the associated graphs of hard instances (to determine if there is a valid assignment) contain some induced $C_4$. In many NP-completeness proofs variants of 1in3SAT (for instance \cite{Mo-Ro}), the reduction adds two clauses that repeat two literals in order to assure that the third literal has the same value in any valid assignment. This action makes the associated graph to contain induced $C_4$'s. 
\\

Below are the six decision problems that we deal with \textcolor{black}{in this article. The first two of them are known NP-Complete variants of 1in3SAT}
\\

PLANAR POSITIVE 1in3SAT \\
INPUT: A positive formula $F$ in CNF, each clause is a disjunction of exactly three different literals and the associated graph $G(F)$ is planar\\
QUESTION: Is there a valid assignment \textcolor{black}{for $F$?}
\\
 
The NP-Completeness proof of this variant can be found in \cite{Laroche,Mu-Ro}.\\

CUBIC PLANAR POSITIVE 1in3SAT \\
INPUT: A positive formula $F$ in CNF, each clause is a disjunction of exactly three different literals and the associated graph $G(F)$ is cubic planar\\
QUESTION: Is there a valid assignment \textcolor{black}{for $F$?}
\\

The NP-Completeness proof of this more restricted variant can be found in \cite{Mo-Ro} where a polynomial reduction is presented to transform an input instance $F$ of PLANAR POSITIVE 1in3SAT %(a positive formula $F$ in CNF, each clause is a disjunction of exactly three different literals and the associated graph $G(F)$ is planar) 
to an instance $F'$ of CUBIC PLANAR POSITIVE 1in3SAT. %(a positive formula $F'$ in CNF, each clause is a disjunction of exactly three different literals and the associated graph $G(F')$ is cubic planar)
It is important to see that if $G(F)$ is connected then $G(F')$ is also connected.
\\

Now, we consider the new connected versions of these two previous 1in3SAT variants.
\\

CONNECTED PLANAR POSITIVE 1in3SAT \\
INPUT: A positive formula $F$ in CNF, each clause is a disjunction of exactly three different literals and the associated graph $G(F)$ is connected and planar\\
QUESTION: Is there a valid assignment \textcolor{black}{for $F$?}
\\

We prove this restricted variant 1in3SAT is also NP-Complete. Let $F$ be an input instance of PLANAR POSITIVE 1in3SAT, we will construct a positive formula $F^*$ in CNF, input instance of CONNECTED PLANAR POSITIVE 1in3SAT, and show that $F$ and $F^*$ are equivalent. This construction is iterative and will stop when the obtained formula is a valid input instance for CONNECTED PLANAR POSITIVE 1in3SAT. We will show in every step, the new formula $F^*$ of the step is a valid input instance of PLANAR POSITIVE 1in3SAT, the number of connected components of $G(F^*)$ has decreased by 1 compared to the previous formula and both formulas are equivalent.
% which means that if one of them is satisfiable if the other does. 
Initially, let $F^*=F$. Repeat the following PROCEDUE while $G(F^*)$ is not connected.\\
\;\\
PROCEDURE
\begin{enumerate}
\item Let $R$ be a planar representation of $G(F^*)$. 
\item Select $G_1$ and $G_2$, two different connected components of $G(F^*)$.
\item As $G(F^*)$ is a bipartite graph and every rectangle (clause) vertex has exactly 3 circle (variable) vertices as neighbors. Therefore, every connected component of $G(F^*)$ has some circle vertex which lies in the external region of $R$. 
Let $v_1$ (corresponds to some variable $x$) and $v_2$ (corresponds to some variable $y$), circle vertex with this property for $G_1$ and for $G_2$, respectively.
\item Add 3 variables ($x',y'$ and $z$) and 2 clauses ($x \vee z \vee x'$ and $y \vee z \vee y'$) to $F^*$.
\end{enumerate}

It is clear that after the execution of the above procedure, the number of connected components of $G(F^*)$ has decreased by 1 since there is a path connecting $v_1$ and $v_2$ passing through vertices corresponding to $x \vee z \vee x'$, $z$ and $y \vee z \vee y'$. It is not hard to see that the new $G(F^*)$ is still planar (the five new vertices and the six new edges can be drawn in the external region of $R$ without edge crossing) and $F^*$ is still a valid input instance of PLANAR POSITIVE 1in3SAT. Now, if the new $F^*$ is satisfiable then the old $F^*$ is satisfiable because all clauses of the old one are clauses of the new one. If the old $F^*$ is satisfiable, we can extend any valid assignment for it to a valid assignment for the new $F^*$ as below: setting $FALSE$ to $z$, the value of $\neg x$ to $x'$ and $\neg y$ to $y'$. With this the proof is completed.\\

CONNECTED CUBIC PLANAR POSITIVE 1in3SAT \\
INPUT: A positive formula $F$ in CNF, each clause is a disjunction of exactly three different literals and the associated graph $G(F)$ is connected and cubic planar\\
QUESTION: Is there a valid assignment \textcolor{black}{for $F$?}
\\

The polynomial reduction given in \cite{Mo-Ro} transforms connected instances of PLANAR POSITIVE 1in3SAT to connected instances of CUBIC PLANAR POSITIVE 1in3SAT. The NP-Completeness of this variant is guaranteed by the NP-Completeness of CONNECTED CUBIC PLANAR POSITIVE 1in3SAT and the correctness of this reduction.\\

The following is our last new proposed variant of 1in3SAT.
\\

CONNECTED SUBCUBIC PLANAR $C_4$-FREE POSITIVE 1in3SAT \\ 
INPUT: A positive formula $F$ in CNF, each clause is a disjunction of exactly three different literals and the associated graph $G(F)$ is connected, planar, $C_4$-free, with degree at most 3\\
QUESTION: Is there a valid assignment for $F$?
\\

Finally, our target problem.
\\

EXISTENCE OF DIM FOR CONNECTED NSF GRAPHS\\
INPUT: Connected NSF graph $G$\\
QUESTION: Does $G$ contain a dominating induced matching?
\\

\textcolor{black}{Next we prove that CONNECTED SUBCUBIC PLANAR $C_4$-FREE POSITIVE 1in3SAT is NP-complete.}

\begin{theorem}\label{subcubic1in3SAT}
\textcolor{black}{Connected subcubic planar $C_4$-free positive 1in3SAT is NP-Complete.}
\end {theorem}

\begin{proof}

\textcolor{black}{Firstly we} describe a polynomial-time reduction from an input instance of Connected cubic planar positive 1in3SAT, a formula $F$, to an input instance of Connected subcubic planar $C_4$-free positive 1in3SAT, a formula $F'$. \textcolor{black}{Secondly,} we show that $F$ and $F'$ are equivalent. This will prove that Connected subcubic planar $C_4$-free positive 1in3SAT is NP-Complete.

The reduction will consider a positive formula $F$ as input (where all literals are positive). The associated graph $G(F)$ is connected cubic planar bipartite. Every variable $y_i$ of $F$, appears exactly in 3 clauses $C_i^1,C_i^2,C_i^3$ as positive literal. We replace $y_i$ by three variables $x_i^1,x_i^2,x_i^3$ and $x_i^j$ replaces $y_i$ as positive literal in $C_i^j$ for $j=1,2,3$. In order to assure that $x_i^1,x_i^2,x_i^3$ have the same value in every valid assignment, we add 8 new clauses and 10 additional variables as follows (see Figure \ref{ReductionFig1}):
\\

For $j=1,2$
\begin{description}
\item $\qquad\qquad x_i^j\vee w_{i,j}^1\vee w_{i,j}^2$
\item $\qquad\qquad w_{i,j}^1\vee w_{i,j}^3\vee w_{i,j}^4$
\item $\qquad\qquad x_i^{j+1}\vee w_{i,j}^4\vee w_{i,j}^5$
\item $\qquad\qquad w_{i,j}^2\vee w_{i,j}^3\vee w_{i,j}^5$ 
\end{description}

%%%Figure{ReductionFig1}
\begin{figure}
\centering
\begin{tabular}{cc}

\begin{tikzpicture}[sq/.style={rectangle, scale=0.75, draw, fill=black!25, minimum size=0.7cm, inner sep=0pt},
c/.style={circle, scale=0.75, draw, minimum size=0.75cm, inner sep=0pt},
line width=1.25pt, xscale=0.6, yscale= 0.7]
\node (0) [c] at (2.7,0){$y_i$};
\node (1) [sq] at (7.25,0){$C^2_i$};
\node (2) [sq] at (0,4.5){$C^1_i$};
\node (3) [sq] at (0,-4.5){$C^3_i$};
\draw[dashed, x radius=3.9, y radius=4.9] (2.7,0) circle;
\draw[] (1) -- (0) -- (2);
\draw[] (0) -- (3);
\draw[] (8.2,0.1) -- (8.5,0.1) -- (8.5,0.3) -- (9,0) -- (8.5,-0.3) -- (8.5,-0.1) -- (8.2,-0.1) -- cycle;
\end{tikzpicture}

&
\begin{tikzpicture}[sq/.style={rectangle, scale=0.75, draw, fill=black!25, minimum size=0.7cm, inner sep=0pt},
c/.style={circle, scale=0.75, draw, minimum size=0.75cm, inner sep=0pt},
line width=1.25pt, xscale=0.6, yscale= 0.7]
\node (1) [sq] at (7.25,0){$C^2_i$};
\node (2) [sq] at (0,4.5){$C^1_i$};
\node (3) [sq] at (0,-4.5){$C^3_i$};
\draw[dashed, x radius=3.9, y radius=4.9] (2.7,0) circle;
\node (4) [c] at (6,0){$x^2_i$};
\node (5) [c] at (3.75,0.75){$w^4_{i,1}$};
\node (6) [c] at (4.75,1.875){$w^5_{i,1}$};
\node (7) [c] at (3.5,1.875){$w^3_{i,1}$};
\node (8) [c] at (2.25,1.875){$w^1_{i,1}$};
\node (9) [c] at (3.25,3){$w^2_{i,1}$};
\node (10) [c] at (1,3.75){$x^1_i$};
\node (Ca) [sq] at (5,0.75){};
\node (Cb) [sq] at (2.5,0.75){};
\node (Cc) [sq] at (4.5,3){};
\node (Cd) [sq] at (2,3){};
\draw[] (1) -- (4) -- (Ca) -- (5) -- (Cb) -- (7) -- (Cc) -- (6) -- (Ca);
\draw[] (Cb) -- (8) -- (Cd) -- (9) -- (Cc);
\draw[] (Cd) -- (10) -- (2);

\node (5') [c] at (3.75,-0.75){$w^2_{i,2}$};
\node (6') [c] at (4.75,-1.875){$w^1_{i,2}$};
\node (7') [c] at (3.5,-1.875){$w^3_{i,2}$};
\node (8') [c] at (2.25,-1.875){$w^5_{i,2}$};
\node (9') [c] at (3.25,-3){$w^4_{i,2}$};
\node (10') [c] at (1,-3.75){$x^3_i$};
\node (Ca') [sq] at (5,-0.75){};
\node (Cb') [sq] at (2.5,-0.75){};
\node (Cc') [sq] at (4.5,-3){};
\node (Cd') [sq] at (2,-3){};
\draw[] (4) -- (Ca') -- (5') -- (Cb') -- (7') -- (Cc') -- (6') -- (Ca');
\draw[] (Cb') -- (8') -- (Cd') -- (9') -- (Cc');
\draw[] (Cd') -- (10') -- (3);
\end{tikzpicture}

\end{tabular}
\caption{\textcolor{black}{Reduction employed in the NP completeness proof}}
\label{ReductionFig1}
\end{figure}
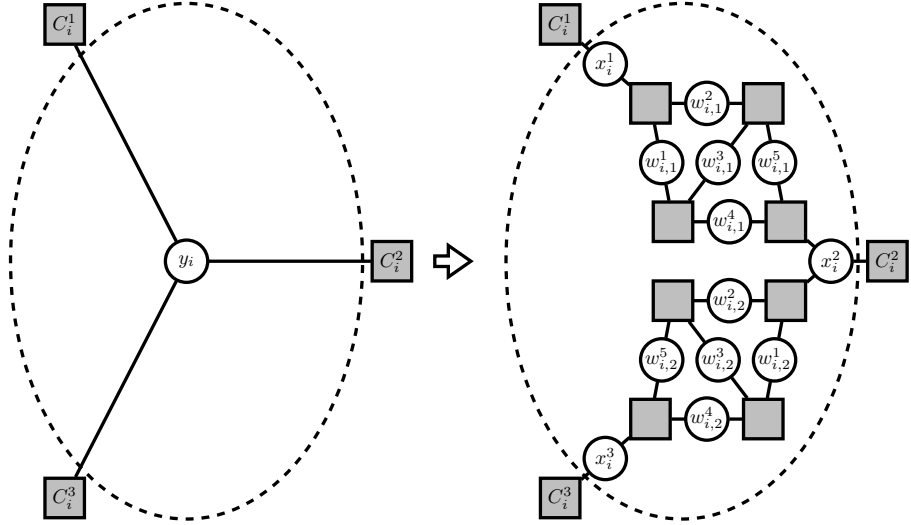
%%%

We show that $w_{i,j}^1,w_{i,j}^5$ have the same value in any valid assignment. If $w_{i,j}^1$ is true, then $x_i^j,w_{i,j}^2,w_{i,j}^3,w_{i,j}^4$ should be false. Therefore, $w_{i,j}^5$ should be the true literal in $w_{i,j}^2\vee w_{i,j}^3\vee w_{i,j}^5$. Symmetrically, if $w_{i,j}^5$ is true, then $x_i^{j+1},w_{i,j}^2,w_{i,j}^3,w_{i,j}^4$ should be false. Therefore, $w_{i,j}^1$ should be the true literal in $w_{i,j}^3\vee w_{i,j}^4\vee w_{i,j}^5$. Now, we prove that $x_i^j,x_i^{j+1}$ have the same value in any valid assignment. In case that $w_{i,j}^1,w_{i,j}^5$ are true, then $x_i^j,x_i^{j+1}$ are both false. Then, we can assume that $w_{i,j}^1,w_{i,j}^5$ are false. If $x_i^j$ is true, then $w_{i,j}^2$ is false and $w_{i,j}^3$ should be the true literal in $w_{i,j}^2\vee w_{i,j}^3\vee w_{i,j}^5$. Consequently, $w_{i,j}^4$ is false and $x_i^{j+1}$ should be the true literal in $x_i^{j+1}\vee w_{i,j}^4\vee w_{i,j}^5$. Symmetrically\textcolor{black}{, if} $x_i^{j+1}$ is true, $x_i^j$ would be true using a similar deduction. \textcolor{black}{As a consequence}, $x_i^1,x_i^2,x_i^3$ have the same value in every valid assignment as we need. Note that all these variables except $x_i^1,x_i^3$ are literals of exactly 2 of these 8 clauses. $x_i^1,x_i^3$ appear only once. But solely, $x_i^1,x_i^2,x_i^3$ appear in exactly one other clause ($C_i^1,C_i^2,C_i^3$, respectively). In conclusion, every variable appears at most in 3 clauses. Clearly, the resulting formula $F'$ is equivalent to $F$. If $F$ uses $n$ variables and $m$ clauses, then $F'$ has exactly $13n$ variables and $m+8n$ clauses. Therefore, $F'$ can be obtained in polynomial time. As all clauses have exactly 3 literals\textcolor{black}{, the associated graph $G(F')$ has} maximum degree 3. In case that $G(F')$ has an induced $C_4$ (see Figure \ref{C4Fig}). Two of the vertices of this $C_4$ correspond to clauses and the other two variables meaning that those two clauses have both variables as literals. This situation does not occur in the formula $F'$. The original $m$ clauses have their local copies of variables. They do not have any variable in common. Any pair of the new $8n$ clauses have at most one variable in common. An original clause and a new one can have at most one variable $x_i^j$ in common. Therefore, $G(F')$ is $C_4$-free. Given a planar representation of $G(F)$, this representation can be easily transformed to \textcolor{black}{a planar representation} of $G(F')$ (see Figure \ref{ReductionFig1}). It is clear that if $G(F)$ is connected then $G(F')$ also does. Finally, we can conclude that $G(F')$ is connected subcubic planar $C_4$-free. Consequently, we can state that Connected subcubic planar $C_4$-free positive 1in3SAT is NP-Complete.
\end{proof}
%%%Figure{C4Fig}
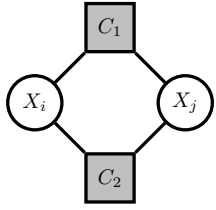
\begin{figure}
\centering
\begin{tikzpicture}[sq/.style={rectangle, scale=0.8, draw, fill=black!25, minimum size=0.8cm, inner sep=0pt},
c/.style={circle, scale=0.8, draw, minimum size=0.9cm, inner sep=0pt},
line width=1.25pt]

\node (0) [sq] at (0,1){$C_1$};
\node (1) [sq] at (0,-1){$C_2$};
\node (2) [c] at (-1,0){$X_i$};
\node (3) [c] at (1,0){$X_j$};
\draw[] (3) -- (0) -- (2) -- (1) -- (3);
\end{tikzpicture}

\caption{\textcolor{black}{An induced $C_4$ cannot exist in the associated graph $G(F')$}}
\label{C4Fig}
\end{figure}
%%%

%We give three alternatives for the gadget and its contact vertex which are used in our reduction.
%\begin{description}
%\item[(a)] $K_1$ and its unique vertex is the contact vertex 
%\item[(b)] butterfly and its central vertex is the contact vertex
%\item[(c)] diamond and choose any degree 2 vertex as the contact vertex
%\end{description}

Now, we give a polynomial-time reduction from an instance of input of Connected subcubic planar $C_4$-free positive 1in3SAT, a positive formula $F$ in CNF, to a graph $S(G(F))$ an instance of our final target problem.

Given such positive formula $F$ where each clause is a disjunction of exactly three different literals and its associate graph $G(F)$ is connected planar $C_4$-free with degree at most 3, we construct a graph $S(G(F))$ from $G(F)$ \textcolor{black}{as follows.}
\begin{enumerate}
\item Replace every rectangle vertex $v$ (corresponding to a clause $C_i$) of $G(F)$ by a triangle (clause triangle). Each vertex of this triangle takes a different neighbor (a circle vertex corresponding to a variable) of $v$ as its owner (see Figure \ref{ReplacingFig1}). Call this resulting graph as \textit{\textbf{temporary}} graph $T(G(F))$.
\item For each circle vertex $w$, add two mutually adjacent new vertices as its neighbors forming a triangle (vertex triangle) (see Figure \ref{ReplacingFig4}). This final graph is the graph $S(G(F))$. 
\end{enumerate}

For our convenience, we call this construction as {\bf Main Transformation}.

Clearly, $S(G(F))$ can be obtained in polynomial-time from $F$. If $F$ has $n$ variables and $m$ clauses, then $S(G(F))$ has exactly $3n+3m$ vertices and $3n+6m$ edges.

\begin{theorem}\label{SGF}
$S(G(F))$ is connected planar NSF ($K_4$,diamond,butterfly, $K_{1,5}$,$H$,\quad snail,press,$C_4$,$C_5$,$C_6$,$C_7$,$C_8$,$C_{3t+1}$,$C_{3t+2}$)-free for $t\geq 3$ without pendant \textcolor{black}{vertices.}
\end {theorem}

\begin{proof}
Clearly, as $G(F)$ is connected then $S(G(F))$ is connected.
Every vertex of $S(G(F))$ is in exactly one triangle which implies that $S(G(F))$ is NSF graph without pendant vertices and it is ($K_4$, diamond, butterfly)-free. 

Every vertex $u$ of a clause triangle $T$ has exactly one neighbor $z\not\in T$ and $z$ is in a vertex triangle ($u$ has degree 3). Every vertex triangle has exactly two simplicial vertices and both of degree two. It is clear that the only vertices that can have degree more than 3 in $S(G(F))$ are those corresponding to circle vertices in $G(F)$. They have degree at most 3 in $G(F)$ but in $S(G(F))$, they have two additional neighbors. Hence, they have degree at most 5 in $S(G(F))$. As every vertex in $S(G(F))$ has at most degree 5, and it is in exactly one triangle, $S(G(F))$ is $K_{1,5}$-free (also, ($K_4$, diamond, butterfly,$K_{1,5}$)-free implies maximum degree 5). 

Next, let us show that $S(G(F))$ is ($H$,snail,press)-free. 

Suppose $S(G(F))$ has a snail \textcolor{black}{as an} induced subgraph. The triangle of the snail should be a vertex triangle since one of its vertices has degree at least 4. But in this case, this triangle has at most one simplicial vertex which is a contradiction. Hence, $S(G(F))$ is snail-free. 

If $S(G(F))$ has an $H$ as an induced subgraph. We consider two cases: (i) the middle edge is part of some triangle $T$ and clearly, $T$ is a vertex triangle, again $T$ has at most one simplicial vertex (ii) the middle edge connects two triangles, one of them is a clause triangle meaning that one of the middle \textcolor{black}{vertices} should have degree exactly 3. All its three \textcolor{black}{neighbors} are in $H$ and there is not any triangle, a contradiction. Therefore, $S(G(F))$ is $H$-free.

Suppose now that $S(G(F))$ contains a press \textcolor{black}{as an} induced subgraph. Clearly, one of the two triangles must be \textcolor{black}{a variable} triangle but both of them have \textcolor{black}{at most} one simplicial vertex. Again, a contradiction. So, $S(G(F))$ is press-free. 

%Next, \textcolor{black}{we analyze} the induced cycles that $S(G(F))$ can have. 
%It is clear that every induced cycle $C_k$ with $k\geq 4$ of $S(G(F))$ is also in $T(G(F))$. $C_k$ should have vertices of clause triangles and circle vertices of $T(G(F))$. The latter follows from the fact that all the neighbors of a circle vertex are vertices of different clause triangles, and 2 vertices of different clause triangles do not have any edge to connect them. Clearly, the intersection of $C_k$ and any clause triangle $T$ contains 0 or 2 vertices, since $C_k$ can not use all the three vertices of $T$. Furthermore, every vertex $v$ of $T$ has exactly one neighbor $w\not\in T$. $C_k$ should reach $T$ using a vertex and leave with another one. Moreover, these two vertices are consecutive in $C_k$. If we contract every pair of vertices of \textcolor{black}{a clause triangles in a} single vertex, then $C_k$ becomes an induced $C_{\frac{2}{3}k}$ which corresponds to an induced cycle of $G(F)$. 

Next, we analyze the possible induced cycles of $S(G(F))$ and $G(F)$.
%\begin{claim}

Let $C$ be an induced cycle of $S(G(F))$ having length greater than 3. We claim that the length of $C$ is reduced by 2/3, regarding its corresponding length in $G(F)$. 
%%\end{claim}
To show this claim, consider a cycle of $G(F)$ and its corresponding cycle of $S(G(F))$. Let $t$ and $k$ represent the lengths of these cycles of $G(F)$ and $S(G(F))$, respectively. Recall that $G(F)$ is planar, $C_4$-free and bipartite, whose parts are formed by circle vertices and rectangle vertices, respectively. 

In $G(F)$ each rectangle vertex has exactly degree 3 because each clause is formed by 3 positive literals. Furthermore, each of these rectangle vertices is replaced by a clause triangle in $T(G(F))$. To extend $T(G(F))$ into $S(G(F))$, we add two new adjacent vertices for each circle vertex $u$, both to be adjacent to $u$. Therefore, the vertices belonging to an induced cycle of length $\geq 4$ in $S(G(F))$ are circle vertices, and vertices of clause triangles. Clearly, they ought to be these both types, since the set of circle vertices forms an independent set, while the clause triangles are disjoint triangles. Clearly, we need at least two circle vertices, to form an induced cycle $C$ of $S(G(F))$. For each such circle vertex, we need two vertices of clause triangles. Therefore, if $C$ contains $t$ circle vertices, $C$ also contains another $2t$ vertices from clause triangles, that is $3t$ vertices in total. Next, every pair of vertices of clause triangles is contracted into a single vertex. Therefore, $t = (2/3)k$, and the length of each cycle of $G(F)$ becomes 2/3 smaller, indeed.

As $G(F)$ is bipartite and $C_4$-free, the only values available for $k$ is $3t$ with $t\geq 3$. Therefore, $S(G(F))$ and $T(G(F))$ are ($C_4$,$C_5$,$C_6$,$C_7$,$C_8$,$C_{3t+1}$,$C_{3t+2}$)-free, for $t\geq 3$. 

As $G(F)$ is planar, take a planar representation of $G(F)$, any vertex is replaced by a triangle. In case of a variable triangle, it is clear that the two adjacent new neighbors of the circle vertex maintain the planar representation (see Figure \ref{ReplacingFig4}). In case of a clause triangle, a similar fact holds (see Figure \ref{ReplacingFig1}). Therefore, $S(G(F))$ is planar.
\end{proof} 

%Depending the alternative of gadget used in the construction, $G$ is 
%\begin{description}
%\item[(a)] neighborhood star-free with some pendant vertices and ($K_4$, diamond, butterfly,$C_5$)-free 
%\item[(b)] neighborhood star-free and ($K_4$, diamond,$C_5$)-free
%\item[(c)] neighborhood star-free and ($K_4$, gem, $W_4$, butterfly,$C_5$)-free
%\end{description}

%%%Figure {ReplacingFig1}
\begin{figure}
\centering
\begin{tabular}{cc}

\begin{tikzpicture}[sq/.style={rectangle, scale=0.8, draw, fill=black!25, minimum size=0.8cm, inner sep=0pt},
c/.style={circle, scale=0.8, draw, minimum size=0.9cm, inner sep=0pt},
line width=1.25pt, scale= 0.8]
\node (0) [sq] at (0,0){$C_i$};
\node (1) [c] at (0,2.6){$X_1$};
\node (2) [c] at (2.6,-1.5){$X_2$};
\node (3) [c] at (-2.6,-1.5){$X_3$};
\draw[dashed, radius=1.8] (0,0) circle;
\draw[] (1) -- (0) -- (2);
\draw[] (0) -- (3);
\draw[] (3.8,0.6) -- (4.1,0.6) -- (4.1,0.8) -- (4.5,0.5) -- (4.1,0.2) -- (4.1,0.4) -- (3.8,0.4) -- cycle;
\end{tikzpicture}
&

\begin{tikzpicture}[sq/.style={rectangle, scale=0.8, draw, fill=black!25, minimum size=0.8cm, inner sep=0pt},
c/.style={circle, scale=0.8, draw, minimum size=0.9cm, inner sep=0pt},
line width=1.25pt, scale= 0.8]
\node (1) [c] at (0,2.6){$X_1$};
\node (2) [c] at (2.6,-1.5){$X_2$};
\node (3) [c] at (-2.6,-1.5){$X_3$};
\draw[dashed, radius=1.8] (0,0) circle;
\node (a) [circle, fill, draw] at (0,1){};
\node (b) [circle, fill, draw] at (0.87,-0.5){};
\node (c) [circle, fill, draw] at (-0.87,-0.5){};
\draw[] (1) -- (a) -- (b) -- (2);
\draw[] (b) -- (c) -- (3);
\draw[] (a) -- (c);
\end{tikzpicture}

\end{tabular}
\caption{Replacing a clause by a triangle}
\label{ReplacingFig1}
\end{figure}
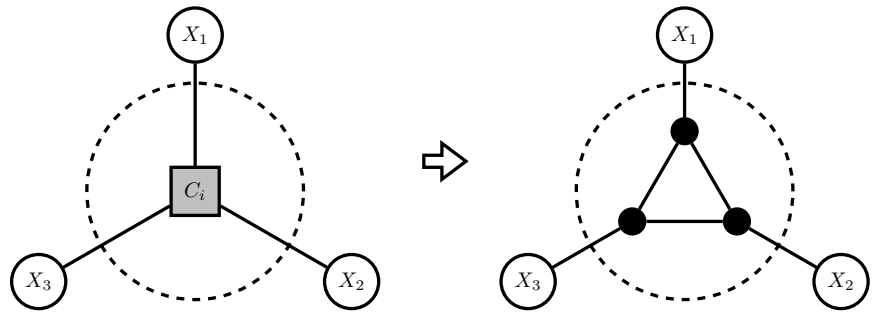
%%%
\begin{figure}
\centering
\begin{tabular}{cc}

\begin{tikzpicture}[sq/.style={rectangle, scale=0.8, draw, fill=black!25, minimum size=0.8cm, inner sep=0pt},
c/.style={circle, scale=0.8, draw, minimum size=0.9cm, inner sep=0pt},
line width=1.25pt, scale= 0.8]
\node (0) [c] at (0,0){$X_i$};
\node (1) [sq] at (0,2.6){$C_1$};
\node (2) [sq] at (2.6,-1.5){$C_2$};
\node (3) [sq] at (-2.6,-1.5){$C_3$};
\draw[dashed, radius=1.8] (0,0) circle;
\draw[] (1) -- (0) -- (2);
\draw[] (0) -- (3);
\draw[] (3.8,0.6) -- (4.1,0.6) -- (4.1,0.8) -- (4.5,0.5) -- (4.1,0.2) -- (4.1,0.4) -- (3.8,0.4) -- cycle;
\end{tikzpicture}
&

\begin{tikzpicture}[sq/.style={rectangle, scale=0.8, draw, fill=black!25, minimum size=0.8cm, inner sep=0pt},
c/.style={circle, scale=0.8, draw, minimum size=0.9cm, inner sep=0pt},
line width=1.25pt, scale= 0.8]
\node (1) [sq] at (0,2.6){$C_1$};
\node (2) [sq] at (2.6,-1.5){$C_2$};
\node (3) [sq] at (-2.6,-1.5){$C_3$};
\draw[dashed, radius=1.8] (0,0) circle;
\node (a) [circle, fill, draw] at (0,0){};
%\node (b) [circle, fill, draw] at (0.87,-0.5){};
\node (c) [circle, fill, draw] at (0.87/2,0.87){};
\node (b) [circle, fill, draw] at (0.87,0){};
%\node (c) [circle, fill, draw] at (-0.87,-0.5){};
\draw[] (1) -- (a) -- (b) -- (c) -- (a);
\draw[] (2) -- (a) -- (3);
\end{tikzpicture}

\end{tabular}
\caption{Replacing a variable by a triangle}
\label{ReplacingFig4}
\end{figure}
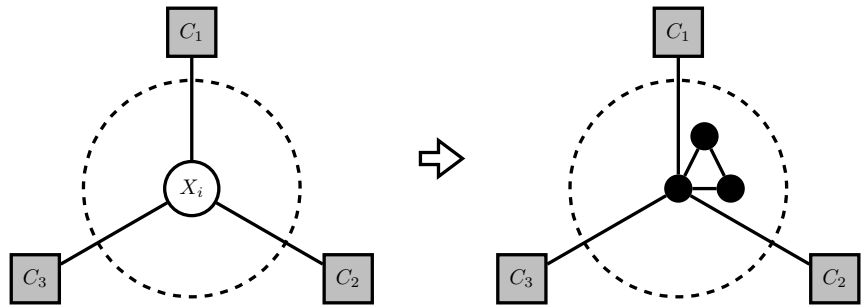
%%%

\begin{theorem}\label{1black}
There is a valid assignment for an input formula $F$ of Connected subcubic planar $C_4$-free positive 1in3SAT if and only if the resulting graph $S(G(F))$ of the main transformation contains some dominating induced matching.
\end {theorem}

\begin{proof}
Given a valid assignment, we define a bi-coloring of $S(G(F))$ using this assignment below.
\begin{enumerate}
\item if a variable $x_i$ is true then \textcolor{black}{color} black its corresponding circle vertex $v$. Choose arbitrarily any neighbor $w$ of its variable triangle and color it also black. Color white all other neighbors of $v$. 
\item if a variable $x_i$ is false then \textcolor{black}{color} white its corresponding circle vertex $v$ and \textcolor{black}{color} black all its neighbors.
\end{enumerate}

\textcolor{black}{Next we} prove that this bi-coloring verifies the conditions of Theorem \ref{check-rulesOri}.
\begin{enumerate}
\item In a variable triangle, every vertex satisfies trivially the conditions.
\item In a clause triangle, a vertex is colored black (white) if its unique circle vertex neighbor $x_i$ is colored white (black) because the variable $x_i$ is false (true). As every clause has one true literal (all literals are positive). In any clause triangle $T$, there is exactly one white vertex $v$ and all its 3 neighbors have been colored black. Particularly, the other two vertices of $T$. 
\end{enumerate}
Therefore, this bi-coloring corresponds to a DIM of $S(G(F))$.\\

Conversely, given a DIM $D$ and its associated bi-coloring, we define an assignment as follows. If a circle vertex $x_i$ is black (white) then the variable $x_i$ is true (false). We will prove this assignment is valid. Using the properties of Theorem \ref{color-rules} and considering $T$ the variable triangle that contains $x_i$, every clause triangle neighbor $w$ of $v$ \textcolor{black}{has a different} color of $x_i$ since the vertices of $T$ and $w$ induce a paw and $x_i,w$ are the odd degree vertices in the induced paw. The color of $w$ is white only if $x_i$ is true. As every \textcolor{black}{clause triangle} has exactly one white vertex. Consequently, every clause has exactly one true literal.
\end{proof}
\begin{corollary}\label{Cor1}
EXISTENCE DIM is NP-Complete for connected planar NSF ($K_4$,\quad diamond,butterfly,$K_{1,5}$,$H$,snail,press,$C_4$,$C_5$,$C_6$,$C_7$,$C_8$,$C_{3t+1}$,$C_{3t+2}$)-free graphs for $t\geq 3$ without pendant vertices.
\end {corollary}

Since connected NSF graphs do not have proper perfect dominating sets, the existence of DIM is equivalent to ask if exist a perfect edge dominating set with at most $m-1$ edges \textcolor{black}{where $|E|=m$} (the trivial perfect dominating set is the set of all \textcolor{black}{edges $E$}).

\begin{corollary}\label{Cor2}
EXISTENCE of a perfect dominating set with at most $m-1$ edges is NP-Complete for Connected planar NSF ($K_4$,diamond,butterfly, $K_{1,5}$,$H$,snail,\quad press,$C_4$,$C_5$,$C_6$,$C_7$,$C_8$,$C_{3t+1}$,$C_{3t+2}$)-free graphs for $t\geq 3$ without pendant vertices where $m$ is the number of edges in the graph.
\end {corollary}
A \textit{\textbf{subdivision}} of an edge $(v,w)$ in a graph consists of replacing such edge by two edges $(v,u),(u,w)$ and a new vertex $u$.
 
\begin{lemma}\label{3subdivision}\cite{CARDOSO2011521} Let $G$ be a graph and $e$ an edge in $G$. If $G'$ is the graph obtained from $G$ by subdividing the edge $e$ exactly three times,
\textcolor{black}{then $G$ has a dominating induced matching if and only if so has $G'$}.
\end{lemma}

Cardoso, et al. \cite {CARDOSO2011521} use the ideas of the triple subdivision of the edges (see Lemma \ref{3subdivision}) to extend the girth of the graph without affecting the existence or not of DIM.
This technique does not work for our class under consideration since this operation adds 3 new vertices of degree 2 and they are not part of triangles. Instead of this operation \textcolor{black}{we propose other} alternatives to avoid induced cycles of size at most $k$ for any $k$. But all of them come with some cost. The main idea is replacing edges of $S(G(F))$, the resulting graph of Main Transformation that connects vertices of different triangles by chains of structures that assure the vertices at both ends of chains should have different colors in any valid bi-coloring. We call those edges as variable-clause edges since they connect a variable triangle to a clause triangle. It is clear that every vertex of the chains should be in a triangle. We define two types of triangles: (i) vertex link triangles, these triangles have only one vertex in the chain, and (ii) edge link triangles, which are those that have exactly two vertices in the chain. There are many alternatives to create chains.
\begin{enumerate}
\item Every chain has exactly $2k_1$ vertex link triangles (see Figure \ref{ReplacingFig2} for $k_1=1$). It is clear that every pair of consecutive vertices in the chain are odd-degree vertices in some induced paw. By Theorem \ref{color-rules}, they have different colors in any valid bi-coloring. Since the chain has $2k_1+1$ edges, then both ends of the chain have different colors in any valid bi-coloring.

%%%Figure {ReplacingFig2}
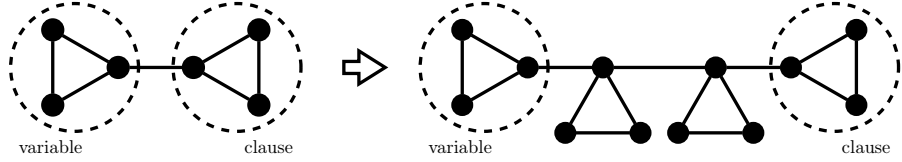
\begin{figure}
\centering
\begin{tabular}{cc}

\begin{tikzpicture}[v/.style={circle, scale=0.8, draw=black, fill},
line width=1.25pt]
\node (0) [v] at (0,0){};
\node (1) [v] at (-0.87,0.5){};
\node (2) [v] at (-0.87,-0.5){};
\draw[dashed, radius=0.85] (-0.58,0) circle;
\node (tx) [scale=0.7] at (-0.9,-1.05){variable};
\node (3) [v] at (1,0){};
\node (4) [v] at (1.87,0.5){};
\node (5) [v] at (1.87,-0.5){};
\draw[dashed, radius=0.85] (1.58,0) circle;
\node (tx) [scale=0.7] at (2,-1.05){clause};
\draw[] (0) -- (1) -- (2) -- (0) -- (3) -- (4) -- (5) -- (3);
\draw[] (3,0.08) -- (3.2,0.08) -- (3.2,0.22) -- (3.55,0) -- (3.2,-0.22) -- (3.2,-0.08) -- (3,-0.08) -- cycle;
\end{tikzpicture}
&

\begin{tikzpicture}[v/.style={circle, scale=0.75, draw=black, fill},
line width=1.25pt]
\node (0) [v] at (0,0){};
\node (1) [v] at (-0.87,0.5){};
\node (2) [v] at (-0.87,-0.5){};
\draw[dashed, radius=0.85] (-0.58,0) circle;
\node (tx) [scale=0.7] at (-0.9,-1.05){variable};
\node (3) [v] at (1+2.5,0){};
\node (4) [v] at (1.87+2.5,0.5){};
\node (5) [v] at (1.87+2.5,-0.5){};
\draw[dashed, radius=0.85] (1.58+2.5,0) circle;
\node (tx) [scale=0.7] at (2+2.5,-1.05){clause};
\node (6) [v] at (1,0){};
\node (7) [v] at (0.5,-0.87){};
\node (8) [v] at (1.5,-0.87){};
\node (9) [v] at (2.5,0){};
\node (10) [v] at (2,-0.87){};
\node (11) [v] at (3,-0.87){};
\draw[] (0) -- (1) -- (2) -- (0) -- (6) -- (9) -- (3) -- (4) -- (5) -- (3);
\draw[] (6) -- (7) -- (8) -- (6);
\draw[] (9) -- (10) -- (11) -- (9);
\end{tikzpicture}

\end{tabular}
\caption{Replacing a variable-clause edge by a chain of 2 vertex link triangles}
\label{ReplacingFig2}
\end{figure}
%%%

\item Every chain has alternately edge link triangles and vertex link triangles if we start from the variable triangle side, the first triangle is an edge link triangle. The number of each type of triangle is $2k_2$ (see Figure \ref{ReplacingFig3} for $k_2=1$). Some additional structure is needed for edge link triangles. Each edge link triangle has exactly one vertex that is not in the chain. We should connect this vertex to some contact vertex of a gadget. There are three options for the gadget structure (see Figure \ref{ReplacingFig3}).
\begin{enumerate}
\item A butterfly, where the contact vertex is its central vertex.
\item A $K_1$ and the contact vertex is the unique one. This vertex \textcolor{black}{will become} a pendant vertex.
\item A diamond, where the contact vertex is any vertex of degree 2 in the diamond.
\end{enumerate}

%%%Figure {ReplacingFig3}
\begin{figure}
\centering
\begin{tabular}{cc}

\begin{tikzpicture}[v/.style={circle, scale=0.75, draw=black, fill},
line width=1.25pt, scale=0.9]
\node (0) [v] at (0,0){};
\node (1) [v] at (-0.87,0.5){};
\node (2) [v] at (-0.87,-0.5){};
\draw[dashed, x radius=0.8, y radius=0.85] (-0.55,0) circle;
\node (tx) [scale=0.7] at (-0.7,-1.05){variable};
\node (3) [v] at (1-0.1,0){};
\node (4) [v] at (1.87-0.1,0.5){};
\node (5) [v] at (1.87-0.1,-0.5){};
\draw[dashed, x radius=0.8, y radius=0.85] (1.55-0.1,0) circle;
\node (tx) [scale=0.7] at (1.6,-1.05){clause};
\draw[] (0) -- (1) -- (2) -- (0) -- (3) -- (4) -- (5) -- (3);
\draw[] (2.8,0.08) -- (3,0.08) -- (3,0.22) -- (3.4,0) -- (3,-0.22) -- (3,-0.08) -- (2.8,-0.08) -- cycle;
\draw[color=white, x radius=0.7, y radius=0.4] (0,-1.9) circle;
\end{tikzpicture}
&

\begin{tikzpicture}[v/.style={circle, scale=0.75, draw=black, fill},
line width=1.25pt, scale=0.9]
\node (0) [v] at (0,0){};
\node (1) [v] at (-0.87,0.5){};
\node (2) [v] at (-0.87,-0.5){};
\draw[dashed, x radius=0.8, y radius=0.85] (-0.55,0) circle;
\node (tx) [scale=0.7] at (-0.7,-1.05){variable};
\node (3) [v] at (1+5.5,0){};
\node (4) [v] at (1.87+5.5,0.5){};
\node (5) [v] at (1.87+5.5,-0.5){};
\draw[dashed, x radius=0.8, y radius=0.85] (1.55+5.5,0) circle;
\node (tx) [scale=0.7] at (1.7+5.5,-1.05){clause};
\node (6) [v] at (1-0.1,0){};
\node (7) [v] at (2-0.1,0){};
\node (8) [v] at (1.5-0.1,-0.87){};
\node (9) [v] at (3-0.2,0){};
\node (10) [v] at (2.5-0.2,-0.87){};
\node (11) [v] at (3.5-0.2,-0.87){};
\node (12) [v] at (4-0.3,0){};
\node (13) [v] at (5-0.3,0){};
\node (14) [v] at (4.5-0.3,-0.87){};
\node (15) [v] at (6-0.4,0){};
\node (16) [v] at (5.5-0.4,-0.87){};
\node (17) [v] at (6.5-0.4,-0.87){};
\filldraw[fill=blue!25, x radius=0.7, y radius=0.4] (1.4,-1.9) circle;
\filldraw[fill=blue!25, x radius=0.7, y radius=0.4] (4.2,-1.9) circle;
\node (a) [circle, draw=black!50, fill=white] at (1.4,-1.87){};
\node (b) [circle, draw=black!50, fill=white] at (4.2,-1.87){};
\draw[] (0) -- (1) -- (2) -- (0) -- (6) -- (7) -- (9) -- (12) -- (13) -- (15) -- (3) -- (4) -- (5) -- (3);
\draw[] (6) -- (8) -- (7);
\draw[] (9) -- (10) -- (11) -- (9);
\draw[] (12) -- (14) -- (13);
\draw[] (15) -- (16) -- (17) -- (15);
\draw[] (a) -- (8);
\draw[] (b) -- (14);
\end{tikzpicture}
\end{tabular}
\vspace{0.25cm}

%%%%%Gadgets
\begin{tikzpicture}[v/.style={circle, scale=0.8, draw=black, fill},
line width=1.25pt]
% Buterfly
\node (0) [circle, draw=black!50] at (0,0.5){};
\node (1) [v] at (-0.87,0){};
\node (2) [v] at (-0.87,1){};
\node (3) [v] at (0.87,0){};
\node (4) [v] at (0.87,1){};
\draw[] (0) -- (1) -- (2) -- (0);
\draw[] (3) -- (0) -- (4) -- (3);
\node (a)[scale=0.9] at (0,1.6){(a) butterfly};
% Pendant
\node (0b) [circle, draw=black!50] at (3.25,0.5){};
\node (a)[scale=0.9] at (3.25,1.6){(b) pendant};
%%%Diamond
\node (0c) [circle, draw=black!50] at (6,0.5){};
\node (1c) [v] at (6.87,1){};
\node (2c) [v] at (6.87,0){};
\node (3c) [v] at (7.74,0.5){};
\node (c)[scale=0.9] at (0.5+6,1.6){(c) diamond};
\draw[] (0c) -- (1c) -- (2c) -- (0c);
\draw[] (1c) -- (3c) -- (2c);
\end{tikzpicture}

\caption{Replacing a variable-clause edge by a mixed chain of 2 vertex link triangles and 2 edge link triangles and 3 options for gadget}
\label{ReplacingFig3}
\end{figure}
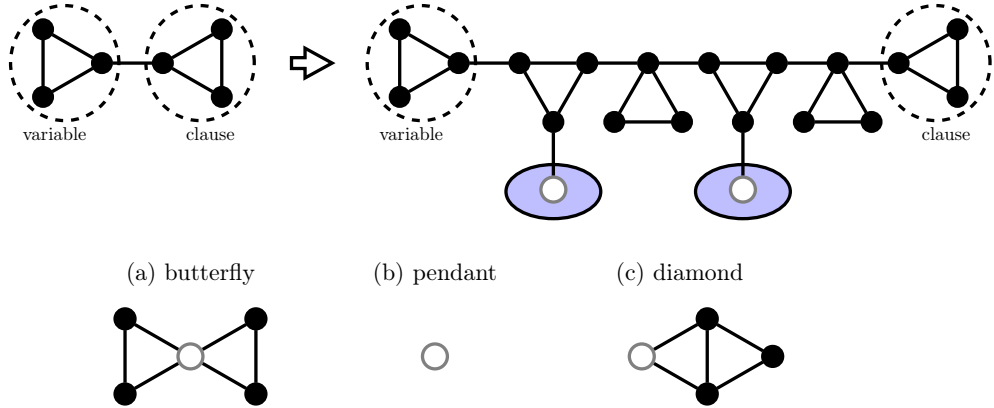
%%%

In any case, the contact vertex is colored white in any valid bi-coloring by Theorem \ref{color-rules} and its neighbor in the edge link triangle should be black, again by Theorem \ref{color-rules}. In consequence, the vertices of the edge link triangle in the chain should have different colors. All other edges of the chain form part of some induced paw and their incident vertices have degree odd in such paw. Therefore, every pair of consecutive vertices in the chain have different colors in any valid bi-coloring. The number of edges in the chain is exactly $6k_2+1$, \textcolor{black}{an odd} number, which implies that both ends of the chain have \textcolor{black}{a different} color in any valid bi-coloring. 
\end{enumerate}
It is clear that $S(G(F))$ admits a DIM if and only if the resulting graph $Q$, after the replacement of variable-clause edges by the chains described above, has a DIM. Furthermore, every variable-clause edge of $S(G(F))$ is part of some induced paw and its incident vertices have odd degree in such induced paw. \textcolor{black}{Consequently, they} should have different colors in any valid bi-coloring, by Theorem \ref{color-rules}. Then from any valid bi-coloring of $Q$, we can obtain trivially a valid bi-coloring of $S(G(F))$. Conversely, any valid bi-coloring of $S(G(F))$, can be extended easily to a bi-coloring of $Q$. Consequently, we can define new reductions from Subcubic planar $C_4$-free positive 1in3SAT to the DIM problem, concatenating the Main Transformation and one of these types of replacement of variable-clause edges. 

Next, let us examine the resulting graph class of each type of replacement. The number of variable-clause edges in an induced cycle in $S(G(F))$ is at least 6 and the number of clause triangle edges is at least 3 since the induced cycles in the associated graph $G(F)$ have at least 3 clause vertices.
\begin{enumerate}
\item In this case, the only forbidden subgraph of Corollaries \ref{Cor1} and \ref{Cor2} that can appear \textcolor{black}{is the $H$ graph}. All induced cycles of $Q$ should have at least \textcolor{black}{$6\cdot(2k_1+1)+3=12k_1+9$ vertices. The number of vertices of $Q$ is $n+(m-n)\cdot 3\cdot 2k_1=6k_1m-(6k_1-1)\cdot n$ and $Q$ has $m+(m-n)\cdot 4\cdot 2k_1=(8k_1+1)\cdot m-8k_1n$ edges.}
\item Independently of \textcolor{black}{the gadget} employed, all induced cycles of $Q$ should have at least \textcolor{black}{$6\cdot (3\cdot 2k_2+1)+3=36k_2+9$} vertices.
\begin{enumerate}
\item In this case, the only forbidden subgraph of Corollaries \ref{Cor1} and \ref{Cor2} that can appear is the butterfly. The number of vertices of $Q$ is \textcolor{black}{$n+(m-n)\cdot(8+3)\cdot 2k_2=22k_2m-(22k_2-1)\cdot n$ and $Q$ has $m+(m-n)\cdot(11+4)\cdot 2k_2=(30k_2+1)\cdot m-30k_2n$ edges.}
\item In the present \textcolor{black}{situation,} $Q$ has pendant vertices. The number of vertices of $Q$ is \textcolor{black}{$n+(m-n)\cdot(4+3)\cdot 2k_2=14k_2m-(14k_2-1)\cdot n$ and $Q$ has $m+(m-n)\cdot(5+4)\cdot 2k_2=(18k_2+1)\cdot m-18k_2n$ edges.}
\item The only forbidden subgraph of Corollaries \ref{Cor1} and \ref{Cor2} that can appear is the diamond. We remark that there are no gems, nor $W_4$'s. The number of vertices of $Q$ is \textcolor{black}{$n+(m-n)\cdot(7+3)\cdot 2k_2=20k_2m-(20k_2-1)\cdot n$ and $Q$ has $m+(m-n)\cdot(10+4)\cdot 2k_2=(28k_2+1)\cdot m-28k_2n$ edges.}
\end{enumerate}
\end{enumerate}
Also, in any case, $Q$ will still be a planar graph of degree at most five. We remark that cycles of size between $4$ and $k$, in $Q$, for fixed $k\geq 4$, are forbidden. For each type of replacement, we can determine the smallest required value of $k_1$ or $k_2$ to assure this.

\begin{theorem}
For any fixed value $k\geq 4$, EXISTENCE DIM and EXISTENCE of PED with at most $m-1$ edges are NP-Complete in the following graph classes.
\begin{itemize}
\item Connected planar NSF ($K_4$,diamond,butterfly,$K_{1,5}$,snail,press,$C_4,\dots,C_k$)-free without pendant \textcolor{black}{vertices.}
\item Connected planar NSF ($K_4$,diamond,$K_{1,5}$,$H$,snail,press, $C_4,\dots,C_k$)-free without pendant \textcolor{black}{vertices.}
\item Connected planar NSF ($K_4$,diamond,butterfly,$K_{1,5}$,$H$,snail,press,$C_4,\dots,C_k$)-\textcolor{black}{free.}
\item Connected planar NSF ($K_4$,gem,$W_4$,butterfly,$K_{1,5}$,$H$,snail,press,$C_4,\dots,C_k$)-free without pendant \textcolor{black}{vertices.}
\end{itemize}
\end {theorem}

\section{Conclusions}

We conclude with an open question.
\bigskip

We wonder if there is some algorithmic relation between efficient and perfect edge domination. More specifically, we remark that there are graph classes which admit polynomial time solutions for solving the efficient edge domination problem while being hard for solving the perfect edge domination problem. However, we ask whether there is some graph class for which there exists a polynomial time algorithm for solving the perfect edge domination problem while being hard for the efficient edge domination problem.

\section*{Acknowledgements} 
%We appreciate the comments of an anonymous reviewer, which significantly helped us improving
%the presentation and clarity of this work.
Min Chih Lin was partially supported by CONICET grant PIP 11220200100084CO and UBACyT Grants 20020190100126BA and 20020170100495BA, Argentina.\\
Abilio Lucena acknowledges the support of CNPq grant 309887/2018-6, Brazil.\\
Nelson Maculan was partially supported by CNPq and COPPETEC Foundation, Brazil.\\
Veronica A. Moyano was partially supported by Universidad Nacional de General Sarmiento grant 30/1135, Argentina.\\
Jayme L. Szwarfiter was partially supported by FAPERJ and CNPq, Brazil.

%\begin{acknowledgements}
%If you'd like to thank anyone, place your comments here
%and remove the percent signs.
%\end{acknowledgements}

% BibTeX users please use one of
%\bibliographystyle{spbasic}      % basic style, author-year citations
%\bibliographystyle{spmpsci}      % mathematics and physical sciences
%\bibliographystyle{spphys}       % APS-like style for physics
%\bibliography{}   % name your BibTeX data base

% Non-BibTeX users please use
%\begin{thebibliography}{}
%
% and use \bibitem to create references. Consult the Instructions
% for authors for reference list style.
%
%\bibitem{RefJ}
% Format for Journal Reference
%Author, Article title, Journal, Volume, page numbers (year)
% Format for books
%\bibitem{RefB}
%Author, Book title, page numbers. Publisher, place (year)

\end{document}